\documentclass{article} 
\usepackage{times} 
\usepackage[round]{natbib} 

\usepackage{bbm}
\usepackage{graphicx}
\usepackage{amsmath,amssymb,amsthm,amsfonts}
\usepackage{url}
\usepackage[colorlinks=true,citecolor=blue,urlcolor=blue,linkcolor=blue]{hyperref} 
\usepackage[dvipsnames]{xcolor} 
\usepackage{booktabs} 
\usepackage{multirow} 
\usepackage[flushleft]{threeparttable} 
\usepackage[parfill]{parskip} 
\usepackage[shortlabels]{enumitem} 
\usepackage{subcaption} 
\usepackage{pifont} 
\usepackage{etoc} 
\usepackage{algorithm} 
\usepackage{algpseudocode}
\algrenewcommand\algorithmicrequire{\textbf{Input:}}
\algrenewcommand\algorithmicensure{\textbf{Output:}}
\usepackage[capitalise]{cleveref} 
\usepackage{xspace} 
\usepackage{makecell} 
\usepackage[most]{tcolorbox} 

\def\E{{\mathbb E}}
\def\P{{\mathbb P}}
\def\R{{\mathbb R}}

\def\cC{{\mathcal C}}

\def\cH{{\mathcal H}}
\def\cI{{\mathcal I}}

\def\cK{{\mathcal K}}

\def\cP{{\mathcal P}}

\def\cX{{\mathcal X}}
\def\cY{{\mathcal Y}}

\newtheorem{theorem}{Theorem}

\newtheorem{proposition}{Proposition}
\newtheorem{remark}{Remark}

\newtheorem{corollary}{Corollary}

\newcommand{\Q}{\mathsf{Q}}
\newcommand{\softmax}{\mathsf{softmax}}

\newcommand{\standard}{\textsc{Standard}\xspace}
\newcommand{\classwise}{\textsc{Classwise}\xspace}
\newcommand{\labelw}{\textsc{Label-weighted}\xspace}
\newcommand{\truncplant}{Pl@ntNet\xspace}
\newcommand{\truncinat}{iNaturalist\xspace}

\newcommand\numberthis{\addtocounter{equation}{1}\tag{\theequation}} 

\usepackage[suppress]{color-edits}
\addauthor{td}{ForestGreen}
\usepackage[margin=1in]{geometry} 

\title{Conformal prediction with macro-coverage guarantees}  

\author{
Aabesh Bhattacharyya$^{1}$\thanks{Equal contribution.}
\quad
Tiffany Ding$^{2}$\footnotemark[1]
\quad
Rina Foygel Barber$^{1}$ \\
\\
$^{1}$Department of Statistics, University of Chicago \\
$^{2}$Department of Statistics, UC Berkeley
}
\date{}

\begin{document}
\maketitle

\begin{abstract}
Prediction sets should have high coverage to be useful, but some coverage notions are more practically relevant than others. In the classification setting, class-conditional coverage requires that the prediction set (i.e., the set of candidate labels for a new test point) must achieve the target accuracy level within each class, which may be challenging to satisfy when many classes are rare and have few calibration points. At the other extreme, marginal coverage requires only that coverage holds on average over the distribution of all classes, which can lead to low-probability labels being essentially ignored. To find a middle ground, recent work has introduced macro-coverage, defined as the unweighted average of class-conditional coverages. Macro-coverage offers a compromise between marginal coverage and class-conditional coverage that is particularly appropriate for long-tailed settings. In this work, we show that label-weighted conformal prediction can be used to produce prediction sets with a finite-sample macro-coverage guarantee, and more generally a guarantee on a family of generalized macro-coverage objectives that aggregate coverage at the level of arbitrary class groupings and take a weighted average. We further characterize the form of the smallest prediction sets satisfying a given generalized macro-coverage objective and propose a corresponding conformal score function. We validate our theoretical results on two large-scale image classification datasets.
\end{abstract}

\section{Introduction}

Conformal prediction provides a principled way to convert the output of a black-box predictive model into a prediction set with a finite-sample coverage guarantee under exchangeability assumptions \citep{vovk2005algorithmic,papadopoulos2002inductive,lei2018distribution,angelopoulos2023conformal}. In this work, we focus on multiclass classification, where the goal is to use features $X \in \cX$ to predict a response $Y$ that takes values in a finite label space $\mathcal{Y}$. Of particular interest is long-tailed classification, where the label space is very large and the class frequencies vary by orders of magnitude. Such long-tailed distributions appear in problems such as species identification and medical diagnosis, and in order for prediction sets to be useful in such settings, they should be small enough to be inspected by a human while still reliably containing the true label, even when the true label belongs to a rare class.

Given a target miscoverage level $\alpha \in [0,1]$, standard split conformal prediction \citep{vovk2005algorithmic} provides a way to construct prediction sets $\mathcal{C}:\mathcal{X}\to 2^\mathcal{Y}$ that satisfy
\[
    \mathbb{P}(Y \in \mathcal{C}(X)) \geq 1-\alpha
\]
under the assumption that the observed data is exchangeable. 
This guarantee holds on average over the population, which can lead to undesirable behavior. 
Writing $p(y)=\mathbb{P}(Y=y)$, we can decompose marginal coverage as
\[
    \mathbb{P}(Y \in \mathcal{C}(X))
    =
    \sum_{y\in\mathcal{Y}} p(y) \cdot \mathbb{P}(Y \in \mathcal{C}(X)\mid Y=y).
\]
For a long-tailed distribution, where $p(y)$ is much larger for some $y$ than others, the marginal guarantee places most of its weight on frequent classes. A method may therefore attain marginal coverage while substantially undercovering tail classes. This phenomenon has been observed in recent work on conformal prediction for long-tailed classification, where standard conformal methods tend to overcover common classes and undercover rare classes \citep{ding2026conformal,liu2026conformal}.

A natural remedy is to seek \emph{class-conditional coverage}:
\[
    \mathbb{P}(Y \in \mathcal{C}(X)\mid Y=y)
    \geq 1-\alpha
    \qquad \text{for all } y\in\mathcal{Y}.
\]

Class-conditional coverage treats every class separately and is therefore attractive in class-imbalanced settings. However, it can be too stringent in the long-tailed regime;
when the calibration set contains only a few examples from a rare class, or none at all, class-specific quantile estimates become overly conservative.
Consequently, classwise conformal methods may produce prediction sets that are very large and hence uninformative
\citep{vovk2012conditional,ding2023class}.
This creates a tension: standard (marginal) conformal prediction is efficient but can fail on rare classes, whereas class-conditional conformal prediction ensures coverage for every class but can be impractical when many classes have limited calibration data.

To balance this tension, \citet{ding2026conformal} introduced \emph{macro-coverage}, which can be thought of as a relaxation of class-conditional coverage. This concept is inspired by macro-accuracy in multiclass classification \citep{lewis1991evaluating}, where ``macro'' refers to the idea of zooming out to the class level before aggregating. Macro-coverage averages the class-conditional coverages uniformly over labels:
\begin{align} \label{eq:macrocov}
    \operatorname{MacroCov}(\mathcal{C})
    :=
    \frac{1}{|\mathcal{Y}|}
    \sum_{y\in\mathcal{Y}}
    \mathbb{P}(Y \in \mathcal{C}(X)\mid Y=y).
\end{align}
Unlike marginal coverage, macro-coverage gives each class equal weight in the coverage criterion, regardless of its prevalence. By aiming to achieve macro-coverage, we can hope to do better on rare classes, without going to the extreme of asking for high class-conditional coverage for each class individually. 

\paragraph{Our contribution.}
In this paper, we propose to use \emph{label-weighted conformal prediction} to produce prediction sets that directly achieve a macro-coverage guarantee (rather than a stronger, but excessively stringent, class-conditional coverage guarantee), which no existing method is able to do.
We prove that label-weighted conformal prediction can be used to not only achieve a macro-coverage guarantee but also to guarantee generalized notions of macro-coverage, where aggregation occurs at the level of an arbitrary grouping of classes and the weights of each group do not have to be equal. Notably, these weights can even be chosen after seeing partial information about the calibration data, allowing us to downweight coverage of groups that do not appear in the calibration data to avoid producing uninformative sets. 
Furthermore, given a desired generalized macro-coverage objective, we derive the form of the smallest sets that satisfy this objective and propose a corresponding conformal score function for approximating these theoretically optimal sets. 

\subsection{Related work}

Class-conditional and group-conditional conformal prediction provide stronger guarantees than marginal conformal prediction by requiring coverage within each class or group \citep{vovk2012conditional}. In multiclass problems with many classes, however, classwise calibration can lead to large sets because rare classes may have very few calibration examples, necessitating large finite-sample adjustments. Clustered conformal prediction addresses this by targeting the data scarcity issue and groups data from classes with similar score distributions \citep{ding2023class}. On the other hand, rank-calibrated class-conditional conformal prediction directly targets set sizes by introducing a rank-based threshold to the set construction, while maintaining class-conditional coverage \citep{shi2024conformal}. However, these methods are not designed to handle settings with extreme class imbalance where many classes have few calibration examples. Whereas previous work aims for class-conditional coverage, our work targets weighted averages of group-conditional coverages, which is a more practical objective in highly class imbalanced settings.

Our work is most directly motivated by conformal prediction for long-tailed classification~\citep{ding2026conformal}, where the authors introduced macro-coverage and derived the prevalence-adjusted softmax score as an approximately optimal score for this objective. Concurrently, \citet{liu2026conformal} proposed tail-aware conformal methods to reduce coverage disparities between head and tail classes. These methods, although motivated by the limitations of marginal coverage in long-tailed settings, still provide marginal rather than macro-coverage guarantees. We instead use label-weighted conformal calibration to obtain finite-sample guarantees for macro-coverage and generalized macro-coverage objectives.

Technically, our label-weighted conformal prediction approach is related to weighted conformal prediction, which has been used to address covariate shift, label shift, and other departures from exchangeability \citep{tibshirani2019conformal,podkopaev2021distribution,barber2022conformal}. We build most directly on group-weighted conformal prediction \citep{bhattacharyya2024group}, adapting its finite-sample arguments from groups defined by features to groups defined by labels. Lastly, our optimal score construction is connected to least ambiguous set-valued classification \citep{sadinle2019least}, but specialized to the generalized macro-coverage objective we introduce.

\section{Method} \label{sec:method}

We first generalize the concept of macro-coverage, as defined in \eqref{eq:macrocov}, then define an algorithm for producing prediction sets guaranteed to have generalized macro-coverage of at least $1-\alpha$ for any user-chosen $\alpha \in [0,1]$.

\emph{Setting.} Let $\alpha \in [0,1]$ be a desired miscoverage level. Consider independent and identically distributed samples $(X_1,Y_1),\dots,(X_n,Y_n), (X_{n+1},Y_{n+1})$, where $\{(X_1,Y_1),\dots,(X_n,Y_n)\}$ form the calibration dataset and $(X_{n+1},Y_{n+1})$ is the test point, with $Y_{n+1}$ unobserved. We operate in the split conformal prediction setting and assume access to a fixed conformal score function $s: \cX \times \cY \to \R$ trained on data that is separate from the calibration samples, where a larger value of $s(x,y)$ indicates that $y$ is less likely to be the label corresponding to $x$.

\subsection{Generalized macro-coverage} \label{sec:generalized_macrocov}

Let $g: \cY \to \cK$ be a grouping function that assigns each class to a group and let $w: \cK \to \R_{\geq 0}$ be a fixed function assigning weights to each group such that $\sum_{k \in \cK} w(k) = 1$.
For any prediction set $\mathcal C: \mathcal{X}\to 2^{\mathcal{Y}}$, we define the \emph{$(g,w)$ macro-coverage} as 
\[
\operatorname{MacroCov}_{g,w}(\mathcal C)
:=
\sum_{k \in \cK}
w(k) \cdot
\mathbb P\bigl(Y_{n+1} \in \mathcal C(X_{n+1})\mid g(Y_{n+1})=k\bigr),
\]
the weighted average of group-conditional coverages. 
More generally, we can choose the weights assigned to each group \emph{after observing partial information about the calibration data}. Formally, denote the group membership of each calibration example as $G_i := g(Y_i)$ and record the realized group memberships of the calibration examples in 
\[
\mathcal H := \sigma(G_1,\dots,G_n).
\]
The weight function $w$ can be any $\cH$-measurable function. For such weight functions, we express the $(g,w)$ macro-coverage as 
\[
\operatorname{MacroCov}_{g,w}(\mathcal C \mid \mathcal H)
:=
\sum_{k \in \cK}
w(k) \cdot
\mathbb P\bigl(Y_{n+1} \in \mathcal C(X_{n+1})\mid g(Y_{n+1})=k,\mathcal H\bigr).
\]
By choosing $g$ and $w$ appropriately, practitioners can prioritize coverage across head and tail classes, clinically meaningful subpopulations, demographic groups, or other pre-specified partitions of interest. 

\paragraph{Examples.} 
We now present some instantiations of generalized macro-coverage. 
\begin{enumerate}[leftmargin=2em]
    \item \emph{Vanilla macro-coverage.} Using the identity group function $g(y) = y$ and uniform weights $w(y) = 1/|\cY|$ yields $\mathrm{MacroCov}$, as defined in \eqref{eq:macrocov}. We will refer to this as simply ``macro-coverage.''
    \item \emph{Tail-focused macro-coverage.} Let $\cY_{\mathrm{tail}} \subset \cY$ denote the classes in the tail of the class distribution (e.g.,  $\cY_{\mathrm{tail}}$ could be the 10\% of classes with the fewest examples). In some settings, it may be especially important to correctly identify these rare classes. For some upweighting factor $\lambda > 1$, the tail-focused macro-coverage (denoted $\mathrm{MacroCov}_{\mathrm{tail}}$) is defined to be the $(g,w)$ macro-coverage with the identity grouping function $g(y) = y$ and weight function 
    \begin{align*}
    w(y) = \begin{cases}
    \frac{\lambda}{W} & \text{if } y \in \cY_{\text{tail}} \\
    \frac{1}{W} & \text{otherwise,}
    \end{cases}
    \end{align*}
    where $W = \lambda |\cY_{\mathrm{tail}}| + |\cY \setminus \cY_{\mathrm{tail}}|$ is a normalization constant. 
    
    \item \emph{Genus-level macro-coverage.} Rather than aggregating at the class level, in settings where classes have a hierarchical structure, we may instead wish to aggregate at a higher level. For example, in plant identification, each class is a plant species that is further grouped into a genus. We get $\mathrm{GenusMacroCov}$ by using the grouping function $g: \cY \to \{1, 2, \dots, M\}$ that groups plant species by genus, where $M$ denotes the total number of genera, and the weight function $w(k) = 1/M$ for $k \in [M]$ that assigns uniform weights to each genus. By targeting $\mathrm{GenusMacroCov}$ rather than $\mathrm{MacroCov}$, we upweight the importance of species that are abundant \emph{relative to other species in the same genus} (and downweight species that are relatively rare within a genus).
    \item \emph{Marginal coverage.} The standard conformal prediction objective of marginal coverage is also an instance of generalized macro-coverage, where $g$ groups all classes into a single group (e.g., $g(y) = 1$ for all $y$), so the only valid weight function is $w(1) = 1$. 
\end{enumerate}

\subsection{Label-weighted conformal prediction}

For each group $k \in \cK$, define $\cI_k = \{i \in [n] : g(Y_i) = k\}$ to be the indices of calibration examples in group $k$ and let $N_k = |\cI_k|$ be the number of calibration examples in group $k$. We index the calibration scores corresponding to group $k$ as $\smash{S^{(k)}_i}$ for $i = 1, 2, \dots, |\cI_k|$.

We then define the \emph{$(g,w)$ label-weighted conformal prediction set} as 
\begin{equation}\label{eq:GWCP_prediction_set}
    \mathcal C(x)
:=
\left\{y\in\mathcal Y:s(x,y)\leq  \Q_{1-\alpha}\left(\sum_{k\in \cK}
w(k) P_k\right)\right\},
\end{equation}
where $\Q_{1-\alpha}(\cP) := \min\{ t \in \R : \P_{V \sim \cP}(V \leq t) \geq 1-\alpha\}$ denotes the $1-\alpha$ quantile of distribution $P$ and
\begin{align*}
    P_k = \begin{cases}
        \frac{1}{N_k}\sum_{j=1}^{N_k}\delta_{S_j^{(k)}} & \text{if $N_k > 0$} \\
        \delta_{\infty} & \text{if $N_k = 0$}.
    \end{cases}
\end{align*}
In words, we compute the $1-\alpha$ quantile of the weighted average of empirical score distributions (where we substitute in a point mass at infinity for groups with no calibration examples), then use this quantile to threshold the score to determine which $y$ values to include in the set. 

\begin{theorem} \label{thm:macrocov_guarantee}
Let $\cK_{+} = \{k \in \cK : N_k > 0\}$ be the set of groups with non-zero calibration examples. Then the prediction set given in \eqref{eq:GWCP_prediction_set} satisfies
$$\operatorname{MacroCov}_{g,w}(\mathcal C \mid \mathcal H) \geq 1-\alpha- \max_{k \in \cK_+} \frac{w(k)}{N_k}.$$

\end{theorem}

Note that conditional on $\mathcal{H}$, $\max_{k \in \cK_+} \frac{w(k)}{N_k}$ is a constant. The proof of this result closely follows the proof of validity for the group-weighted conformal prediction approach proposed in \citet{bhattacharyya2024group}, but adapted to label-based groupings. The detailed proof is given in Appendix~\ref{app:proof_validity}. The fact that the guarantee holds conditional on $\mathcal{H}$ makes it stronger than the unconditional macro-coverage guarantee.

This theorem directly implies a procedure for producing prediction sets with a $(g,w)$ macro-coverage guarantee of $1-\alpha$, which we describe in Algorithm \ref{alg:lwcp}. In brief, we must simply compute the correction factor $\Delta := \max_{k \in \cK_+} \frac{w(k)}{N_k}$, then construct the label-weighted conformal prediction set at a corrected miscoverage level $\alpha' = \alpha - \Delta$. 

\begin{corollary}
Let $\cC$ be the prediction set produced by Algorithm~\ref{alg:lwcp} with target miscoverage level $\alpha \in [0,1]$. Then
$
\operatorname{MacroCov}_{g,w}(\cC \mid \cH) \geq 1-\alpha$.
\end{corollary}

\begin{algorithm}[h] 
\caption{Label-weighted conformal prediction for guaranteed $(g,w)$ macro-coverage}
\label{alg:lwcp}
\begin{algorithmic}[1]
\Require Calibration data $\{(X_i, Y_i)\}_{i=1}^n$, score function $s$, grouping function $g: \cY \to \cK$, group weights $w: \cK \to \R_{\geq 0}$ with $\sum_{k \in \cK} w(k) = 1$, target miscoverage $\alpha \in [0,1]$, test point $X_{n+1}$
\Ensure Prediction set $\cC(X_{n+1})$ with $(g,w)$ macro-coverage of at least $1-\alpha$ 
\State Compute calibration scores $S_i \gets s(X_i, Y_i)$ and groups $G_i \gets g(Y_i)$ for $i \in [n]$
\For{each group $k \in \cK$}
    \State $\cI_k \gets \{i \in [n] : G_i = k\}$
    \State $N_k \gets |\cI_k|$
    \State $P_k \gets \frac{1}{N_k}\sum_{i \in \cI_k}\delta_{S_i}$ if $N_k > 0$, else $P_k \gets \delta_{\infty}$
\EndFor
\State $\cK_+ \gets \{k \in \cK : N_k > 0\}$
\State $\alpha' \gets \alpha - \max_{k \in \cK_+} \frac{w(k)}{N_k}$
\State $\cC(X_{n+1}) \gets \{y \in \cY : s(X_{n+1}, y) \leq \Q_{1 - \alpha'}\!\left(\sum_{k \in \cK} w(k)\, P_k\right)\}$ if $\alpha' \geq 0$, else $\cC(X_{n+1}) = \cY$
\end{algorithmic}
\end{algorithm}

When implementing label-weighted conformal prediction, it can be useful to recast it as vanilla weighted conformal prediction, in which a weight is assigned to each calibration point. Let $\cK_0 := \{k \in \cK: N_k = 0\}$ be the groups that have no calibration examples. Then \eqref{eq:GWCP_prediction_set} is equivalent to 
\begin{equation}
    \mathcal C(x) =
\left\{y\in\mathcal Y:s(x,y)\leq  \Q_{1-\alpha}\left(\sum_{i=1}^n
w_i \delta_{S_i} + \sum_{k \in \mathcal{K}_0} w(k) \delta_{\infty}\right)\right\},
\end{equation}
where the weight of calibration point $(X_i, Y_i)$ is $w_i := w(g(Y_i)) / N_{g(Y_i)}$. Note that the calibration point weights do not sum to one ($\sum_{i=1}^n w_i \neq 1$) if there are groups $k \in \cK$ that have zero calibration examples. Thus, the $\sum_{k \in \mathcal{K}_0} w(k) \delta_{\infty}$ term ensures that we are taking the quantile of a valid probability distribution, with a total mass of one. From this, it is also easy to see that if the total weight of the missing groups exceeds $\alpha$, the quantile will be infinite, resulting in infinite sets.

\subsection{Optimal score function for generalized macro-coverage}\label{sec:optimal_score}

When we construct prediction sets, we generally have two objectives. First, we want the sets to satisfy our coverage criterion. Second, we want the sets to be as informative as possible, subject to the coverage constraint. In this section, we focus on the case where the weights $w(k)$ are fixed and do not depend on the calibration data. We propose a way to approximately achieve the theoretically optimal sets, where we use ``optimal'' to mean minimizing expected set size subject to the chosen generalized macro-coverage objective. 
Let $p(y) = \P(Y = y)$ be the probability of class $y$ and define $\rho(k) = \sum_{y : g(y) = k} p(y)$ for $k \in \cK$ to be the probability that a label belongs to group $k$. The following proposition characterizes the optimal prediction set that satisfies a chosen generalized macro-coverage constraint.

\begin{proposition} [informal]\label{prop:theoretical_optimal_set}
    The solution to 
    \begin{align*}
        \min_{\cC: \cX \to 2^{\cY}} \E|\cC(X)| \quad \text{subject to } \mathrm{MacroCov}_{g,w}(\cC) \geq 1-\alpha
    \end{align*}
    is of the form 
    \begin{align} \label{eq:theoretical_optimal_set}
        \mathcal{C}^*_t(x) = \left\{y \in \mathcal{Y} : \frac{w(g(y))}{\rho(g(y))
        } \cdot p(y|x) \geq t\right\} \quad \text{for some threshold } t.
    \end{align}

\end{proposition}
This result is obtained by showing that the $(g,w)$ macro-coverage constraint can be expressed as $(g_I, \tilde w)$ macro-coverage for the identity mapping $g_I(y) = y$ and weight function
\begin{align} \label{eq:induced_macrocov_weights}
    \tilde w(y) := \frac{w(g(y)) \cdot p(y)}{\rho(g(y))},
\end{align}
then applying Proposition 6 of \cite{ding2026conformal}. The formal statement of the result and its proof is given in Appendix \ref{app:proof_theoretical_optimal_set}.

Observe that we can rewrite \eqref{eq:theoretical_optimal_set} as
\begin{align*}
    \mathcal{C}^*_t(x) =  \left\{y \in \mathcal{Y} : s^*_{g,w}(x,y) \leq t \right\} \quad \text{for } s^*_{g,w}(x,y) = -\frac{w(g(y))}{\rho(g(y))
    } \cdot p(y|x).
\end{align*}
Thus, combining Proposition \ref{prop:theoretical_optimal_set} and Theorem \ref{thm:macrocov_guarantee}, we know that if we were able to use $s^*_{g,w}$ as our score function and construct our prediction set according to Algorithm \ref{alg:lwcp}, this would yield the smallest set with $(g,w)$ macro-coverage of at least $1-\alpha$. However, in practice, $s^*_{g,w}$ cannot be used since it depends on unknown population-level probabilities, so we instead substitute in plug-in estimates, yielding the score function
\begin{align} \label{eq:approx_optimal_score}
    \hat s_{g,w}(x,y) = -\frac{w(g(y))}{\widehat \rho(g(y))} \cdot \widehat p(y|x).
\end{align}
We can get $\widehat p(y|x)$ by training any probabilistic classifier. $\widehat \rho (g(y))$ can be obtained using the label distribution in the classifier training dataset.
Since Theorem \ref{thm:macrocov_guarantee} holds for any score function, it holds for $\hat s_{g,w}(x,y)$, so applying Algorithm \ref{alg:lwcp} with this score function yields the desired generalized macro-coverage guarantee while approximating the smallest set achieving this guarantee. 

\subsection{Simultaneous coverage guarantees} \label{sec:simultaneous_guarantees}

In some applications, we want prediction sets that satisfy multiple coverage guarantees simultaneously. 
For example, we might want to guarantee that our sets have high macro-coverage \emph{and} high marginal coverage. Formally, consider $m \geq 1$ constraints:  for $j=1,2,\dots, m$, $g_j$ is a grouping function, $\cH_j$ records the group membership of the calibration examples (as determined by $g_j$), $w_j$ is an $\cH_j$-measurable weighting function, and $\alpha_j$ is a miscoverage level. Our goal is to generate prediction sets that simultaneously satisfy 
\begin{align} \label{eq:simultaneous_coverage_goal}
\mathrm{MacroCov}_{g_j, w_j}(\cC \mid \cH_j) \geq 1 - \alpha_j \quad \text{for } j = 1, 2, \dots, m
\end{align}

\begin{proposition} \label{prop:simultaneous_cov_via_max}
Fix a score function $s$ and run Algorithm \ref{alg:lwcp} for each condition $j=1,2,\dots, m$ and extract $\widehat q_j := \Q_{1 - \alpha_j'}\!\left(\sum_{k \in \cK_j} w_j(k)\, P_k^j\right)$ from line $9$, where $\cK_j$ are the groups given by $g_j$ and $P_k^j$ and
$\alpha_j'$ are as computed by Algorithm~\ref{alg:lwcp}. Then the prediction set
\begin{align}
\bar \cC(x) := \{y \in \cY : s(x,y) \leq \max(\widehat q_1, \widehat q_2, \dots, \widehat q_m)\}
\end{align}
satisfies \eqref{eq:simultaneous_coverage_goal}.
\end{proposition} 

Proofs for results in this section are given in Appendix~\ref{app:oracle_score_simultaneous_cov}. 
Proposition \ref{prop:simultaneous_cov_via_max} says that we can produce sets that satisfy multiple generalized macro-coverage conditions by simply taking the max of their score thresholds. 
As with the single coverage constraint problem, we can motivate a choice of score function for the multiple constraint setting by deriving the form of the optimal prediction set. 

\begin{proposition}[informal]\label{prop:optimal_score_simultaneous_coverage}
The solution to 
    \begin{align*}
        \min_{\cC: \cX \to 2^{\cY}} \E|\cC(X)| \quad \mathrm{ s.t. } \quad \mathrm{MacroCov}_{g_j,w_j}(\cC) \geq 1-\alpha_j \quad \text{for all }j=1,2,\dots, m
    \end{align*}
    is of the form 
    \begin{align} \label{eq:theoretical_simultaneous_optimal_set}
        \mathcal{C}^*_t(x) = \left\{y \in \mathcal{Y} : 
\frac{p(y|x)}{p(y)} \cdot \sum_{j =1}^m \lambda_j \tilde w_j(y)\geq t\right\} \quad \text{for some  $\lambda_1, \lambda_2, \dots, \lambda_m$ and $t$},
    \end{align}
    where $\tilde w_j : \cY \to \R_{\geq 0}$ is the class-level weighting function induced by $(g_j, w_j)$, as defined in ~\eqref{eq:induced_macrocov_weights}.
\end{proposition}

We state a formal version of this proposition in
Appendix~\ref{app:oracle_score_simultaneous_cov} and provide the proof there.
In practice, to approximate the oracle set in
\eqref{eq:theoretical_simultaneous_optimal_set}, we can use plug-in estimates of
$p(y|x)$ and $p(y)$, as in Section~\ref{sec:optimal_score}. The additional
difficulty in the simultaneous-coverage setting is that the score also depends on
the vector $\lambda:= (\lambda_1,\ldots,\lambda_m)$.
However, in many applications, only a subset of the coverage constraints is active. For example, if we 
target both marginal coverage and macro-coverage at the same miscoverage level $\alpha$ in long-tailed classification problems, it is often the case that the macro-coverage constraint is the more demanding one. This because it places more emphasis on covering rare classes, which classifiers struggle more with.  In such cases, a procedure that attains the desired macro-coverage
level may also attain the marginal coverage target, so that the marginal coverage
constraint is inactive in the corresponding oracle problem. Note that for any inactive constraints, the corresponding $\lambda_j$ in Proposition~\ref{prop:optimal_score_simultaneous_coverage} will be zero. However, in general, we can set $\lambda$ using data separate from the calibration data (e.g., held-out data or the data used for classifier training) to do grid search over potential values and selecting the value that gives the smallest estimated average set size. 



\section{Experiments}

We perform experiments using real data to demonstrate that label-weighted conformal prediction achieves the macro-coverage objectives guaranteed by our theoretical results. We use the softmax score function, given by $s_{\softmax}(x,y) = -\hat p(y|x)$, which is also called the Least Ambiguous Classifier (LAC) score \citep{sadinle2019least}. We obtain $\hat p(y|x)$ from a fine-tuned ResNet-50 classifier initialized to ImageNet-pretrained weights and use the resulting softmax scores, following the procedure described in \cite{ding2026conformal}. 
We also use the approximately optimal score function $\hat s_{g,w}$ defined in \eqref{eq:approx_optimal_score}. Our primary evaluation metrics are macro-coverage ($\mathrm{MacroCov}$) and average set size ($\mathrm{AvgSize}$). We also compute marginal coverage ($\mathrm{MarginalCov}$). Code is available at \url{https://github.com/tiffanyding/macro-guarantees}.

\paragraph{Methods.} We compare three methods. The first baseline method is \standard conformal prediction, which constructs prediction sets as $$\mathcal{C}(x) = \left\{y : s(x,y) \leq \Q_{1-\alpha}\left(\frac{1}{n+1} \sum_{i=1}^n \delta_{S_i} + \frac{1}{n+1}\delta_{\infty}\right)\right\}$$ and guarantees $1-\alpha$ marginal coverage. The second baseline is \classwise conformal prediction, which constructs prediction sets as $$\mathcal{C}(x) = \left\{y : s(x,y) \leq \Q_{1-\alpha}\left(\frac{1}{N_y+1} \sum_{i=1}^{N_y} \delta_{S_i^{(y)}} + \frac{1}{N_y+1}\delta_{\infty}\right)\right\},$$ where $N_y$ and $S_i^{(y)}$ are defined analogously to $N_k$ and $S_i^{(k)}$ in Section \ref{sec:method}, for the special case where $\cK = \cY$ (each class is its own group). \classwise guarantees $1-\alpha$ class-conditional coverage for all classes. We compare these two baselines against \labelw conformal prediction, as described in Algorithm \ref{alg:lwcp}, for group and weight functions $(g,w)$ described below.

\paragraph{Data.} We use two large-scale image classification datasets: the plant species dataset Pl@ntNet-300K \citep{garcin2021pl} and the 2018 version of iNaturalist \citep{van2015building}, which contains images of plants, animals, insects, and more. We choose these datasets as a test bed due to their high class imbalance, which is the setting in which macro-coverage is most relevant. 
To ensure reliable evaluation of macro-coverage, which requires computing class-conditional coverage, we use the truncated versions of these datasets created in \cite{ding2026conformal} by filtering out classes with too few examples for testing. 
This leaves 330 Pl@ntNet classes and 857 iNaturalist classes. 
Pl@ntNet is highly skewed, with the most common class being 100 times more common than the rarest class. iNaturalist is less skewed, with a ratio of 10. 
To compute standard errors, we combine the provided calibration and test splits (resulting in 98,061 total examples for \truncplant and 50,906 for \truncinat) and do our own random splitting into calibration and test datasets, where each example is assigned to the calibration dataset with probability 0.1. All metrics are computed using 20 random seeds. 


\subsection{Main results: Achieving macro-coverage}

The goal of this experiment is to construct prediction sets with $\mathrm{MacroCov} \geq 1-\alpha$ for a user-chosen $\alpha \in [0,1]$. 
To target this goal, we run \labelw with the identity grouping function $g(y) = y$ and uniform weights $w(y) = 1/|\cY|$ for all $y$. 
Note that for this $g$ and $w$, the score $\hat s_{g,w}$ is equivalent to the 
prevalence-adjusted softmax score $s(x,y) = - \hat p(y|x) / \hat p(y)$ proposed in \cite{ding2026conformal} for approximating the optimal score function for targeting macro-coverage.

As a high-level summary, this experiment empirically validates our two main theoretical results: First, that label-weighted conformal prediction can be used to achieve $1-\alpha$ macro-coverage for a user-chosen $\alpha$, and second, that $\hat s_{g,w}$ achieves the smallest set size at a given macro-coverage level. 
Tables~\ref{tab:plantnet-trunc_main}--\ref{tab:inaturalist-trunc_main}  show the result of running the various conformal prediction methods on \truncplant and \truncinat at levels $\alpha=0.05$ and $\alpha=0.1$. In the $\mathrm{MacroCov}$ column, we bold entries with macro-coverage at least $1-\alpha$, or within one standard error of it. In the $\mathrm{AvgSize}$ column, we bold the minimum average size among rows achieving the desired macro-coverage. We report the standard errors as subscripts. We use this presentation style throughout the rest of the paper.

We observe that $\standard$, which guarantees $1-\alpha$ marginal coverage, does not in general achieve $1-\alpha$ macro-coverage. 
$\classwise$ results are only reported for $s_{\softmax}$ because $\classwise$ is invariant to class-dependent adjustments to score functions, so $s_{\softmax}$ and $\hat s_{g,w}$ produce identical results (for any $g$ and $w$). 
$\classwise$ does achieve $1-\alpha$ macro-coverage, but it overshoots. 
Although overshooting the coverage target is not itself undesirable, here it comes at the cost of extremely large prediction sets.
The cause of this overshooting is that $\classwise$ guarantees $1-\alpha$ class-conditional coverage, which is stronger than the macro-coverage condition we ask for. By contrast, $\labelw$ is able to achieve the desired macro-coverage without overshooting. Furthermore, comparing the rows with bolded $\mathrm{MacroCov}$ entries (denoting that $1-\alpha$ macro-coverage is achieved), we see that \labelw~with $\hat s_{g,w}$ has the smallest average set size, supporting the theoretical results in Section \ref{sec:optimal_score}.

\begin{table}[h]
\centering
\setlength{\tabcolsep}{2.3pt}
\caption{Targeting macro-coverage on \truncplant}\label{tab:plantnet-trunc_main}
\begin{tabular}{ll rrr rrr}
\toprule
 & & \multicolumn{3}{c}{\tcbox[on line, colback=blue!10, colframe=blue!10, arc=3pt, boxsep=1pt, left=3pt, right=3pt, top=1pt, bottom=1pt]{ $\mathsf{Goal}: \mathrm{MacroCov} \geq 0.9$}} & \multicolumn{3}{c}{\tcbox[on line, colback=blue!10, colframe=blue!10, arc=3pt, boxsep=1pt, left=3pt, right=3pt, top=1pt, bottom=1pt]{$\mathsf{Goal}: \mathrm{MacroCov} \geq 0.95$}} \\
\cmidrule(lr){3-5} \cmidrule(lr){6-8}
Method & Score & \small{$\mathrm{MarginalCov}$} & \small{$\mathrm{MacroCov}$} & \small{$\mathrm{AvgSize}$} & \small{$\mathrm{MarginalCov}$} & \small{$\mathrm{MacroCov}$} & \small{$\mathrm{AvgSize}$} \\
\midrule
\multirow{2}{*}{\standard} & $s_{\softmax}$ & $0.901_{\pm 0.001}$ & $0.861_{\pm 0.001}$ & $2.8_{\pm 0.0}$ & $0.951_{\pm 0.001}$ & $0.929_{\pm 0.001}$ & $5.9_{\pm 0.1}$ \\
 & $\hat s_{g,w}$ & $0.900_{\pm 0.001}$ & $0.886_{\pm 0.001}$ & $2.2_{\pm 0.0}$ & $0.950_{\pm 0.001}$ & $0.939_{\pm 0.001}$ & $3.7_{\pm 0.0}$ \\
\midrule
\classwise & $s_{\softmax}$ & $0.940_{\pm 0.001}$ & $\mathbf{0.940}_{\pm 0.001}$ & $58.1_{\pm 1.4}$ & $0.987_{\pm 0.000}$ & $\mathbf{0.992}_{\pm 0.000}$ & $263.4_{\pm 0.9}$ \\
\midrule
\multirow{2}{*}{\makecell[l]{\textsc{Label-}\\\textsc{weighted}}} & $s_{\softmax}$ & $0.930_{\pm 0.001}$ & $\mathbf{0.901}_{\pm 0.001}$ & $4.0_{\pm 0.0}$ & $0.966_{\pm 0.001}$ & $\mathbf{0.951}_{\pm 0.001}$ & $9.2_{\pm 0.1}$ \\
 & $\hat s_{g,w}$ & $0.915_{\pm 0.001}$ & $\mathbf{0.901}_{\pm 0.001}$ & $\mathbf{2.4}_{\pm 0.0}$ & $0.959_{\pm 0.000}$ & $\mathbf{0.949}_{\pm 0.001}$ & $\mathbf{4.4}_{\pm 0.0}$ \\
\bottomrule
\end{tabular}
\end{table}

\begin{table}[h]
\centering
\setlength{\tabcolsep}{2.3pt}
\caption{Targeting macro-coverage on \truncinat}\label{tab:inaturalist-trunc_main}
\begin{tabular}{ll rrr rrr}
\toprule
 & & \multicolumn{3}{c}{\tcbox[on line, colback=blue!10, colframe=blue!10, arc=3pt, boxsep=1pt, left=3pt, right=3pt, top=1pt, bottom=1pt]{ $\mathsf{Goal}: \mathrm{MacroCov} \geq 0.9$}} & \multicolumn{3}{c}{\tcbox[on line, colback=blue!10, colframe=blue!10, arc=3pt, boxsep=1pt, left=3pt, right=3pt, top=1pt, bottom=1pt]{$\mathsf{Goal}: \mathrm{MacroCov} \geq 0.95$}} \\
\cmidrule(lr){3-5} \cmidrule(lr){6-8}
Method & Score & \small{$\mathrm{MarginalCov}$} & \small{$\mathrm{MacroCov}$} & \small{$\mathrm{AvgSize}$} & \small{$\mathrm{MarginalCov}$} & \small{$\mathrm{MacroCov}$} & \small{$\mathrm{AvgSize}$} \\
\midrule
\multirow{2}{*}{\standard} & $s_{\softmax}$ & $0.901_{\pm 0.001}$ & $0.894_{\pm 0.001}$ & $10.1_{\pm 0.1}$ & $0.950_{\pm 0.001}$ & $0.946_{\pm 0.001}$ & $26.9_{\pm 0.3}$ \\
 & $\hat s_{g,w}$ & $0.901_{\pm 0.001}$ & $0.898_{\pm 0.001}$ & $7.2_{\pm 0.1}$ & $0.950_{\pm 0.001}$ & $\mathbf{0.949}_{\pm 0.001}$ & $\mathbf{18.6}_{\pm 0.3}$ \\
\midrule
\classwise & $s_{\softmax}$ & $0.941_{\pm 0.001}$ & $\mathbf{0.940}_{\pm 0.001}$ & $206.4_{\pm 3.0}$ & $0.998_{\pm 0.000}$ & $\mathbf{0.998}_{\pm 0.000}$ & $826.9_{\pm 1.4}$ \\
\midrule
\multirow{2}{*}{\makecell[l]{\textsc{Label-}\\\textsc{weighted}}} & $s_{\softmax}$ & $0.908_{\pm 0.001}$ & $\mathbf{0.901}_{\pm 0.001}$ & $11.3_{\pm 0.1}$ & $0.955_{\pm 0.001}$ & $\mathbf{0.951}_{\pm 0.001}$ & $30.2_{\pm 0.4}$ \\
 & $\hat s_{g,w}$ & $0.903_{\pm 0.001}$ & $\mathbf{0.901}_{\pm 0.001}$ & $\mathbf{7.5}_{\pm 0.1}$ & $0.953_{\pm 0.000}$ & $\mathbf{0.951}_{\pm 0.001}$ & $19.6_{\pm 0.2}$ \\
\bottomrule
\end{tabular}
\end{table}


$\standard$ achieves closer to the desired macro-coverage on \truncinat~compared to \truncplant. This is likely because the class distribution of \truncinat, although imbalanced, is closer to uniform than \truncplant. Recall that for a perfectly uniform class distribution, $\mathrm{MarginalCov}$ equals $\mathrm{MacroCov}$, so $\standard$, by targeting $\mathrm{MarginalCov}$, can also indirectly do well on $\mathrm{MacroCov}$ when the class distribution is not too imbalanced.

\subsection{Results beyond macro-coverage} \label{sec:results_beyond_macrocov}

We now demonstrate that label-weighted conformal prediction can be used to produce prediction sets that satisfy generalized macro-coverage conditions. 
We consider two notions of generalized macro-coverage motivated by plant identification. 

\paragraph{(a) Tail-focused macro-coverage.} For biodiversity monitoring purposes, it is particularly important to correctly identify rare plant species. Motivated by this application, in this experiment, we seek to produce \truncplant prediction sets with high tail-focused macro-coverage, as defined in Section \ref{sec:generalized_macrocov}. That is, we seek $\mathrm{MacroCov}_{\mathrm{tail}} \geq 1-\alpha$. To construct $\cY_{\mathrm{tail}}$, we take the 33 classes (10\% of the 330 total classes) with the fewest examples. We upweight coverage on these classes by $\lambda = 10$. 

The left panel of Table \ref{tab:plantnet-trunc_generalized_alpha=0.1} shows results for $\alpha = 0.1$. We observe that $\labelw$ successfully achieves the desired $1-\alpha$ tail-focused macro-coverage level and that applying the optimal $\hat s_{g,w}$ score results in the smallest average set size, by many factors, compared to the commonly used softmax score.  
Results for $\alpha = 0.05$ are qualitatively similar (see Appendix \ref{sec:additional_experiments_APPENDIX}).

\paragraph{(b) Genus-level macro-coverage.} If we are more interested in equalizing coverage at the genus level rather than the species level, we can aim to construct \truncplant prediction sets satisfying $\mathrm{GenusMacroCov} \geq 1-\alpha$ (see Section \ref{sec:generalized_macrocov}). The right panel of Table \ref{tab:plantnet-trunc_generalized_alpha=0.1} shows the results for $\alpha = 0.1$, with $\alpha=0.05$ results deferred to Appendix \ref{sec:additional_experiments_APPENDIX}.
We observe that $\labelw$ achieves the desired $\mathrm{GenusMacroCov}$ level.
$\standard$ with the $\hat{s}_{g,w}$ score also manages to achieve $1-\alpha$ $\mathrm{GenusMacroCov}$, but it overshoots slightly, resulting in a suboptimal average set size. Conversely, $\labelw$ achieves $1-\alpha$ $\mathrm{GenusMacroCov}$ almost exactly and yields the smallest average set size.

\begin{table}[h]
\centering
\setlength{\tabcolsep}{2.2pt}
\caption{Targeting generalized macro-coverage on \truncplant\ at $\alpha = 0.1$. The left panel targets tail-focused macro-coverage; the right panel targets genus-level macro-coverage.}\label{tab:plantnet-trunc_generalized_alpha=0.1}
\begin{tabular}{ll rrr rrr}
\toprule
 & & \multicolumn{3}{c}{\tcbox[on line, colback=green!10, colframe=green!10, arc=3pt, boxsep=1pt, left=3pt, right=3pt, top=1pt, bottom=1pt]{$\mathsf{Goal}: \mathrm{MacroCov}_{\mathrm{tail}} \geq 0.9$}} & \multicolumn{3}{c}{\tcbox[on line, colback=orange!15, colframe=orange!15, arc=3pt, boxsep=1pt, left=3pt, right=3pt, top=1pt, bottom=1pt]{$\mathsf{Goal}: \mathrm{GenusMacroCov} \geq 0.9$}} \\
\cmidrule(lr){3-5} \cmidrule(lr){6-8}
Method & Score & \small{$\mathrm{MarginalCov}$} & \small{$\mathrm{MacroCov}_{\mathrm{tail}}$} & \small{$\mathrm{AvgSize}$} & \small{$\mathrm{MarginalCov}$} & \tiny{$\mathrm{GenusMacroCov}$} & \small{$\mathrm{AvgSize}$} \\
\midrule
\multirow{2}{*}{\standard} & $s_{\softmax}$ & $0.901_{\pm 0.001}$ & $0.656_{\pm 0.002}$ & $2.8_{\pm 0.0}$ & $0.901_{\pm 0.001}$ & $0.883_{\pm 0.001}$ & $2.8_{\pm 0.0}$ \\
 & $\hat{s}_{g,w}$ & $0.900_{\pm 0.001}$ & $0.849_{\pm 0.001}$ & $2.3_{\pm 0.0}$ & $0.900_{\pm 0.001}$ & $\mathbf{0.941}_{\pm 0.000}$ & $2.8_{\pm 0.0}$ \\
\midrule
\classwise & $s_{\softmax}$ & $0.940_{\pm 0.001}$ & $\mathbf{0.945}_{\pm 0.001}$ & $58.1_{\pm 1.4}$ & $0.941_{\pm 0.001}$ & $\mathbf{0.949}_{\pm 0.001}$ & $10.7_{\pm 0.2}$ \\
\midrule
\multirow{2}{*}{\makecell[l]{\textsc{Label-}\\\textsc{weighted}}} & $s_{\softmax}$ & $0.984_{\pm 0.000}$ & $\mathbf{0.901}_{\pm 0.002}$ & $24.3_{\pm 0.9}$ & $0.919_{\pm 0.001}$ & $\mathbf{0.902}_{\pm 0.001}$ & $3.4_{\pm 0.0}$ \\
 & $\hat{s}_{g,w}$ & $0.949_{\pm 0.001}$ & $\mathbf{0.899}_{\pm 0.002}$ & $\mathbf{3.8}_{\pm 0.1}$ & $0.764_{\pm 0.002}$ & $\mathbf{0.902}_{\pm 0.001}$ & $\mathbf{1.8}_{\pm 0.0}$ \\
\bottomrule
\end{tabular}
\end{table}

\section{Discussion}

We propose label-weighted conformal prediction as a way to produce prediction sets with finite-sample guarantees of (generalized) macro-coverage. 
Rather than defaulting to the marginal guarantee of standard conformal, which implicitly assumes that the importance of covering a class $y$ is given by its prevalence, our work expands the scope of distribution-free guarantees that are achievable in classification settings.
Using our framework, practitioners can precisely specify how much coverage for each class matters, define an aggregated coverage notion that respects the relative importances, then produce prediction sets with a guarantee that their customized coverage is at least $1-\alpha$. Furthermore, using our score function proposal means that these sets will be approximately optimal, in the sense of set size.

\emph{Limitations and future work.}
Conformal prediction methods, including ours, have an inherent tradeoff between coverage and set size, which worsens as data becomes more limited. 
As the number of calibration examples of the smallest class (or, more generally, group) shrinks, the size of the $\alpha$ correction in Algorithm \ref{alg:lwcp} grows, causing prediction set size to increase with it. However, the amount of increase is controlled by the weight assigned to the smallest group. Since our theory allows us to adjust the weight assigned to a group after seeing its calibration count, this allows us to avoid producing uninformative sets. 
Furthermore, the adjustment our method applies is favorable compared to that required by class-conditional conformal. Whereas the latter applies an adjustment for \emph{every} class, our method only applies an adjustment for the ``worst class'' (the class with the highest weight, normalized by the number of calibration examples).
Directions for future work include
exploring practical ways of simultaneously achieving guaranteed macro-coverage and feature-conditional coverage and
operationalizing our framework for real-world classification tasks, such as the plant identification app from which the \truncplant dataset used in our experiments is derived.

\bibliographystyle{plainnat}
\bibliography{ref}  

\newpage
\appendix

\section{Proof of Theorem~\ref{thm:macrocov_guarantee}}\label{app:proof_validity}

\begin{proof}
Conditional on $\mathcal H$,
the group memberships $g(Y_1),\ldots,g(Y_n)$, the counts $N_k$, the set
$\mathcal K_+$, and the weights $w(k)$ are fixed.

Let
\[
    \Delta
    :=
    \max_{k\in\mathcal K_+}\frac{w(k)}{N_k}.
\]
If $1-\alpha-\Delta\le 0$, then the result is immediate because coverage is
nonnegative. Thus, assume $1-\alpha-\Delta>0$.

We first prove the result in the case where all groups are observed, so that
$\mathcal K_+=\mathcal K$.

For each $k\in\mathcal K$, let
\[
    (X_0^{(k)},Y_0^{(k)})
\]
be a``ghost" sample from the distribution of $(X,Y)\mid g(Y) = k$ that is independent of the calibration data,
and define
\[
    S_0^{(k)} := s(X_0^{(k)},Y_0^{(k)}).
\]
Since $\mathcal K_+=\mathcal K$, the weighted empirical score distribution is
\[
    \widehat P_{g,w}
    =
    \sum_{k\in\mathcal K}
    w(k)\frac{1}{N_k}\sum_{j=1}^{N_k}\delta_{S_j^{(k)}}.
\]
By the definition of $(g,w)$ macro-coverage,
\[
    \operatorname{MacroCov}_{g,w}(\mathcal C\mid\mathcal H)
    =
    \sum_{k\in\mathcal K}
    w(k)
    \mathbb P
    \left(
        S_0^{(k)}
        \le
        \Q_{1-\alpha}(\widehat P_{g,w})
        \mid
        \mathcal H
    \right).
\]

Fix $k\in\mathcal K$ and $j\in\{1,\ldots,N_k\}$. Define the ghost-replaced
weighted empirical distribution
\[
    \widehat P_{g,w}^{\,k,j}
    :=
    \widehat P_{g,w}
    +
    \frac{w(k)}{N_k}
    \left(
        \delta_{S_0^{(k)}}-\delta_{S_j^{(k)}}
    \right).
\]
That is, $\widehat P_{g,w}^{\,k,j}$ is obtained from $\widehat P_{g,w}$ by replacing
the $j$th calibration score in group $k$ by the independent ghost score from the
same group.

Conditional on $\mathcal H$, the scores
\[
    S_0^{(k)},S_1^{(k)},\ldots,S_{N_k}^{(k)}
\]
are exchangeable. Therefore,
\[
    \mathbb P
    \left(
        S_0^{(k)}
        \le
        \Q_{1-\alpha}(\widehat P_{g,w})
        \mid
        \mathcal H
    \right)
    =
    \mathbb P
    \left(
        S_j^{(k)}
        \le
        \Q_{1-\alpha}(\widehat P_{g,w}^{\,k,j})
        \mid
        \mathcal H
    \right).
\]
Averaging over $j=1,\ldots,N_k$ and summing over $k\in\mathcal K$ gives
\begin{align}
    \operatorname{MacroCov}_{g,w}(\mathcal C\mid\mathcal H)
    &=
    \sum_{k\in\mathcal K}
    w(k)\frac{1}{N_k}
    \sum_{j=1}^{N_k}
    \mathbb P
    \left(
        S_j^{(k)}
        \le
        \Q_{1-\alpha}(\widehat P_{g,w}^{\,k,j})
        \mid
        \mathcal H
    \right).
\label{eq:all_observed_exchangeability_step}
\end{align}

For every $k\in\mathcal K$ and $j\in\{1,\ldots,N_k\}$, the distributions
$\widehat P_{g,w}$ and $\widehat P_{g,w}^{\,k,j}$ differ only by moving mass
$w(k)/N_k$ from $S_j^{(k)}$ to $S_0^{(k)}$. Hence
\[
    d_{\mathrm{TV}}
    \left(
        \widehat P_{g,w},
        \widehat P_{g,w}^{\,k,j}
    \right)
    \le
    \frac{w(k)}{N_k}
    \le
    \Delta.
\]
Therefore,
\[
    \Q_{1-\alpha}(\widehat P_{g,w}^{\,k,j})
    \ge
    \Q_{1-\alpha-\Delta}(\widehat P_{g,w}).
\]
Using this in \eqref{eq:all_observed_exchangeability_step},
\begin{align*}
    \operatorname{MacroCov}_{g,w}(\mathcal C\mid\mathcal H)
    &\ge
    \sum_{k\in\mathcal K}
    w(k)\frac{1}{N_k}
    \sum_{j=1}^{N_k}
    \mathbb P
    \left(
        S_j^{(k)}
        \le
        \Q_{1-\alpha-\Delta}(\widehat P_{g,w})
        \mid
        \mathcal H
    \right) \\
    &=
    \mathbb E
    \left[
    \sum_{k\in\mathcal K}
    w(k)\frac{1}{N_k}
    \sum_{j=1}^{N_k}
    \mathbf 1
    \left\{
        S_j^{(k)}
        \le
        \Q_{1-\alpha-\Delta}(\widehat P_{g,w})
    \right\}
    \,\middle|\,
    \mathcal H
    \right].
\end{align*}
By the definition of the weighted empirical distribution and of the quantile,
\[
    \sum_{k\in\mathcal K}
    w(k)\frac{1}{N_k}
    \sum_{j=1}^{N_k}
    \mathbf 1
    \left\{
        S_j^{(k)}
        \le
        \Q_{1-\alpha-\Delta}(\widehat P_{g,w})
    \right\}
    \ge
    1-\alpha-\Delta
\]
almost surely. Therefore, in the all-observed case,
\[
    \operatorname{MacroCov}_{g,w}(\mathcal C\mid\mathcal H)
    \ge
    1-\alpha-\Delta.
\]

We now handle the general case, where some groups may have $N_k=0$.
Let
\[
    W_+ := \sum_{k\in\mathcal K_+} w(k)
\]
be the total weight assigned to observed groups.

If $\Q_{1-\alpha}(\widehat P_{g,w})=\infty$, then  $\mathcal \cC(x)=\mathcal Y$ for all $x\in\mathcal X$
and the macro-coverage guarantee is satisfied trivially. 

If $\Q_{1-\alpha}(\widehat P_{g,w})<\infty$, then 
we must have $W_+>0$. Define normalized weights on the observed groups by
\[
    \widetilde w(k)
    :=
    \frac{w(k)}{W_+}
    \qquad
    \text{for } k\in\mathcal K_+,
\]
and define the observed-groups weighted empirical score distribution
\[
    \widehat P_{+}
    :=
    \sum_{k\in\mathcal K_+}
    \widetilde w(k)
    \frac{1}{N_k}
    \sum_{j=1}^{N_k}
    \delta_{S_j^{(k)}}.
\]
By construction,
\[
    \widehat P_{g,w}
    =
    W_+\widehat P_{+}
    +
    (1-W_+)\delta_\infty \quad \text{ and } \quad 
    \Q_{1-\alpha}(\widehat P_{g,w})
    =
    \Q_{\frac{1-\alpha}{W_+}}(\widehat P_{+}).
\]
Let
\[
    1-\widetilde\alpha
    :=
    \frac{1-\alpha}{W_+},
    \qquad
    \text{equivalently}
    \qquad
    \widetilde\alpha
    =
    1-\frac{1-\alpha}{W_+}.
\]
Applying the all-observed result above to the group set $\mathcal K_+$, the weights
$\widetilde w(k)$, and the miscoverage level $\widetilde\alpha$, we obtain
\[
    \sum_{k\in\mathcal K_+}
    \widetilde w(k)
    \mathbb P
    \left(
        S_0^{(k)}
        \le
        \Q_{\frac{1-\alpha}{W_+}}(\widehat P_{+})
        \mid
        \mathcal H
    \right)
    \ge
    \frac{1-\alpha}{W_+}
    -
    \max_{k\in\mathcal K_+}
    \frac{\widetilde w(k)}{N_k}.
\]
Using
\[
    Q_{\frac{1-\alpha}{W_+}}(\widehat P_+)
    =
    \Q_{1-\alpha}(\widehat P_{g,w}),
\]
and multiplying both sides by $W_+$ gives
\[
    \sum_{k\in\mathcal K_+}
    w(k)
    \mathbb P
    \left(
        S_0^{(k)}
        \le
        \Q_{1-\alpha}(\widehat P_{g,w})
        \mid
        \mathcal H
    \right)
    \ge
    1-\alpha
    -
    W_+ \cdot
    \max_{k\in\mathcal K_+}
    \frac{\widetilde w(k)}{N_k}.
\]
Since
\[
    W_+ \cdot
    \max_{k\in\mathcal K_+}
    \frac{\widetilde w(k)}{N_k}
    =
    \max_{k\in\mathcal K_+}
    \frac{w(k)}{N_k}
    =
    \Delta,
\]
we have
\[
    \sum_{k\in\mathcal K_+}
    w(k)
    \mathbb P
    \left(
        S_0^{(k)}
        \le
        \Q_{1-\alpha}(\widehat P_{g,w})
        \mid
        \mathcal H
    \right)
    \ge
    1-\alpha-\Delta.
\]
Finally,
\begin{align*}
    \operatorname{MacroCov}_{g,w}(\mathcal C\mid\mathcal H)
    &=
    \sum_{k\in\mathcal K}
    w(k)
    \mathbb P
    \left(
        S_0^{(k)}
        \le
        \Q_{1-\alpha}(\widehat P_{g,w})
        \mid
        \mathcal H
    \right) \\
    &\ge
    \sum_{k\in\mathcal K_+}
    w(k)
    \mathbb P
    \left(
        S_0^{(k)}
        \le
        \Q_{1-\alpha}(\widehat P_{g,w})
        \mid
        \mathcal H
    \right) \\
    &\ge
    1-\alpha-\Delta.
\end{align*}
This proves the theorem.
\end{proof}

\section{Proof of Proposition~\ref{prop:theoretical_optimal_set}} \label{app:proof_theoretical_optimal_set}

\begin{proposition}[Formal version of Proposition \ref{prop:theoretical_optimal_set}]
For $t \in \R$, define 
\begin{align}\label{al:optimal_set_macro}
    \tilde \cC_{t}(x) = \left\{ y \in \cY : \tilde{s}(x,y) \geq t\right\} \quad \text{where}\quad  \tilde{s}(x,y)=\frac{w(g(y))}{\rho(g(y))} \cdot p(y|x).
\end{align}
 If there exists $t_{\alpha}$ such that $\mathrm{MacroCov}_{g,w}(\tilde \cC_{t_{\alpha}}) = 1-\alpha$, then $\tilde \cC_{t_{\alpha}}$ is the solution to 
    \begin{align} 
        \min_{\cC : \mathcal{X} \mapsto 2^\mathcal{Y}} \E[|\cC(X)|] \quad \mathrm{s.t.} \quad \mathrm{MacroCov}_{g,w}(\cC) \geq 1-\alpha.
    \end{align}

\begin{remark}
    If there does not exist
    $t_\alpha$ satisfying the equality exactly, the optimal set is still of the thresholded form above but must be combined with randomization to achieve the optimal solution with exact $(g,w)$ macro-coverage of $1-\alpha$. \citet[Theorem~6.1]{shao2008mathematical} can be used to characterize such a randomized solution. 
\end{remark}
    
\end{proposition}

\begin{proof}
    First, observe that for any grouping function $g$ and weight function $w$, we can rewrite $\mathrm{MacroCov}_{g,w}$ in terms of the identity grouping function $g_{I}(y) = y$ and some weighting function $\tilde w: \cY \to \R_{\geq 0}$:
\begin{align*}
\operatorname{MacroCov}_{g,w}(\mathcal C\mid \mathcal H)
&=
\sum_{k \in \cK}
w(k)\,\mathbb P\bigl(Y_{n+1}\in \mathcal C(X_{n+1})\mid g(Y_{n+1})=k,\mathcal H\bigr) \\
&=
\sum_{k\in \cK}
w(k)
\sum_{y : g(y) = k}
\frac{p(y)}{\rho(k)}\,
\mathbb P\bigl(Y\in \mathcal C(X)\mid Y=y,\mathcal H\bigr) \\
&=
\sum_{y\in\mathcal Y}
\tilde w(y) \,
\mathbb P\bigl(Y\in \mathcal C(X)\mid Y=y,\mathcal H\bigr)  \numberthis \label{eq:gen_macrocov_as_weighted} \\
& \qquad \text{for } \quad \tilde w(y) := \frac{w(g(y)) \cdot p(y)}{\rho(g(y))}.
\end{align*}

In order to find the smallest prediction set $\cC$ satisfying $\operatorname{MacroCov}_{g,w}(\mathcal C\mid \mathcal{H}) \geq 1-\alpha$, we use a result from \cite{ding2026conformal}. We restate the relevant part here for completeness.

\emph{Proposition 6 of \cite{ding2026conformal}.}
Let $\omega: \cY \to [0,1]$ be a non-negative weighting function summing to one. For $t \in \R$, define 
\begin{align}\label{al:optimal_set_macro}
    \tilde \cC_{t}(x) = \left\{ y \in \cY : \tilde{s}(x,y) \geq t\right\} \quad \text{where}\quad  \tilde{s}(x,y)=\frac{\omega(y)}{p(y)} \cdot p(y|x)
\end{align}
and $p(y|x)$ denotes the conditional probability of $Y$ given $X=x$ and $p(y)$ is the marginal probability of $Y$.
Let $\alpha \in [0,1]$. If there exists $t_{\alpha}$ such that $\mathrm{MacroCov}_\omega(\tilde \cC_{t_{\alpha}}) = 1-\alpha$, then $\tilde \cC_{t_{\alpha}}$ is the optimal solution to 
    \begin{align} 
        \min_{\cC : \mathcal{X} \mapsto 2^\mathcal{Y}} \E[|\cC(X)|] \quad \text{subject to } \sum_{y \in \mathcal{Y}} \omega(y) \P( y \in \cC(X) \mid Y=y ) \geq 1-\alpha.
    \end{align}

The representation of $(g,w)$ macro-coverage given in \eqref{eq:gen_macrocov_as_weighted} allows us to apply this proposition, which directly yields the desired result.
\end{proof}

\section{Proofs of Propositions~\ref{prop:simultaneous_cov_via_max} and \ref{prop:optimal_score_simultaneous_coverage}}
\label{app:oracle_score_simultaneous_cov}

\begin{proof}[Proof of Proposition~\ref{prop:simultaneous_cov_via_max}]
Let $\cC_j$ refer to the prediction set produced by Algorithm \ref{alg:lwcp} for the $j$-th coverage constraint. By Theorem \ref{thm:macrocov_guarantee}, each $\cC_j$ satisfies $\mathrm{MacroCov}_{g_j, w_j} \geq 1-\alpha_j$. When the true label is in any of these sets, it is also in the union of such sets, so the set union simultaneously satisfies all of the coverage conditions. Because all $\cC_j$ are constructed by applying the same fixed score function $s$ and comparing against a threshold, the set union can be obtained by taking the max of the score thresholds implied by each constraint. 
\end{proof}

Next, we first state the formal version of Proposition~\ref{prop:optimal_score_simultaneous_coverage} and then prove it.

\begin{proposition}[Formal version of Proposition~\ref{prop:optimal_score_simultaneous_coverage}]
\label{prop:formal_optimal_score_simultaneous_coverage}
Suppose there exist \(\lambda=(\lambda_1,\ldots,\lambda_m)\in
\mathbb R_{\ge 0}^m\), not identically zero, and \(t>0\) such that the
deterministic prediction set
\[
    \mathcal C_{\lambda,t}(x)
    =
    \left\{
    y\in\mathcal Y:
    \frac{p(y\mid x)}{p(y)}
    \sum_{j=1}^m \lambda_j\tilde w_j(y)
    \ge t
    \right\}
\]
is feasible, meaning that
\[
    \mathrm{MacroCov}_{g_j,w_j}(\mathcal C_{\lambda,t})
    \ge 1-\alpha_j,
    \qquad j=1,\ldots,m,
\]
and satisfies
\begin{equation}\label{eq:slackness_condition}
    \lambda_j
    \left[
        \mathrm{MacroCov}_{g_j,w_j}(\mathcal C_{\lambda,t})
        -
        (1-\alpha_j)
    \right]
    =
    0,
    \qquad j=1,\ldots,m.
\end{equation}

Then \(\mathcal C_{\lambda,t}\) is an optimal solution to
\[
    \min_{\mathcal C:\mathcal X\to 2^{\mathcal Y}}
    \mathbb E|\mathcal C(X)|
    \quad
    \mathrm{s.t.}
    \quad
    \mathrm{MacroCov}_{g_j,w_j}(\mathcal C)
    \ge 1-\alpha_j,
    \qquad j=1,\ldots,m.
\]
\end{proposition}

\begin{remark}
The proposition is stated for deterministic prediction sets. If the condition above is not satisfied by any $\mathcal{C}_{\lambda,t}$, then an optimal relaxed solution may
require randomization on the boundary in order to satisfy the active
constraints exactly.
\end{remark}

\begin{proof}
Let $c(x,y) = \mathbf{1}\{y \in \mathcal{C}(x)\}$. We prove the result by considering the more general relaxed problem in
which
\[
    c:\mathcal X\times\mathcal Y\to[0,1]
\]
is allowed to be a randomized set-inclusion rule.
Deterministic prediction sets are
recovered when $c$ takes values in $\{0,1\}$. 
Fix $j\in[m]$. For each group $k\in\mathcal K_j$, let $\rho_j(k)=\mathbb P(g_j(Y)=k)$.

For a relaxed rule $c$, the $j$th generalized macro-coverage functional is
\[
    \mathrm{MacroCov}_{g_j,w_j}(\mathcal{C})
    =
    \sum_{k\in\mathcal K_j}
    w_j(k)
    \mathbb E[c(X,Y)\mid g_j(Y)=k].
\]
Expanding over labels inside each group gives
\begin{align*}
    \mathrm{MacroCov}_{g_j,w_j}(c)
    &=
    \sum_{k\in\mathcal K_j}
    w_j(k)
    \sum_{y: g_j(y)=k}
    \mathbb P(Y=y\mid g_j(Y)=k)
    \mathbb E[c(X,y)\mid Y=y] \\
    &=
    \sum_{k\in\mathcal K_j}
    w_j(k)
    \sum_{y:g_j(y)=k}
    \frac{p(y)}{\rho_j(k)}
    \mathbb E[c(X,y)\mid Y=y] \\
    &=
    \sum_{y\in\mathcal Y}
    \frac{w_j(g_j(y))p(y)}{\rho_j(g_j(y))}
    \mathbb E[c(X,y)\mid Y=y] \\
    &=
    \sum_{y\in\mathcal Y}
    \tilde w_j(y)
    \mathbb E[c(X,y)\mid Y=y].
\end{align*}
Using Bayes' rule,
\[
    p(x\mid y)
    =
    \frac{p(y\mid x)p(x)}{p(y)},
\]
we can rewrite this as
\begin{align}
    \mathrm{MacroCov}_{g_j,w_j}(\mathcal{C})
    &=
    \sum_{y\in\mathcal Y}
    \tilde w_j(y)
    \int c(x,y)p(x\mid y)\,dx \nonumber\\
    &=
    \int p(x)
    \sum_{y\in\mathcal Y}
    c(x,y)
    \tilde w_j(y)
    \frac{p(y\mid x)}{p(y)}
    \,dx.
\label{eq:coverage_linear_form}
\end{align}
Similarly, the expected size is
\begin{equation}
\label{eq:size_linear_form}
    \mathbb E|\mathcal C(X)|
    =
    \mathbb E\left[\sum_{y\in\mathcal Y}c(X,y)\right]
    =
    \int p(x)\sum_{y\in\mathcal Y}c(x,y)\,dx.
\end{equation}

Thus the relaxed oracle problem is the linear program
\[
    \min_{0\le c\le 1}
    \int p(x)\sum_{y\in\mathcal Y}c(x,y)\,dx
\quad \text{s.t.}
    \int p(x)
    \sum_{y\in\mathcal Y}
    c(x,y)
    \tilde w_j(y)
    \frac{p(y\mid x)}{p(y)}
    \,dx
    \ge
    1-\alpha_j,
    \quad
    j=1,\ldots,m.
\]

For multipliers \(\mu_1,\ldots,\mu_m\ge 0\), the Lagrangian is
\begin{align*}
    \mathcal L(c,\mu)
    &=
    \mathbb E|\mathcal C(X)|
    +
    \sum_{j=1}^m
    \mu_j
    \left[
        (1-\alpha_j)
        -
        \mathrm{MacroCov}_{g_j,w_j}(c)
    \right] \\
    &=
    \sum_{j=1}^m
    \mu_j(1-\alpha_j)
    +
    \int p(x)
    \sum_{y\in\mathcal Y}
    c(x,y)
    \left[
        1
        -
        \frac{p(y\mid x)}{p(y)}
        \sum_{j=1}^m
        \mu_j\tilde w_j(y)
    \right]
    dx .
\end{align*}

For fixed \(\mu\), the Lagrangian is separable in \((x,y)\). Therefore,
minimizing it over \(c(x,y)\in[0,1]\) can be done pointwise. The coefficient
of \(c(x,y)\) is
\[
    1
    -
    \frac{p(y\mid x)}{p(y)}
    \sum_{j=1}^m
    \mu_j\tilde w_j(y).
\]
Hence any minimizer of the Lagrangian satisfies
\begin{equation}\label{eq:lagrangain_soln_form}
    c_\mu(x,y)
    =
    \begin{cases}
    1, &
    \displaystyle
    \frac{p(y\mid x)}{p(y)}
    \sum_{j=1}^m
    \mu_j\tilde w_j(y)>1,\\[1em]
    0, &
    \displaystyle
    \frac{p(y\mid x)}{p(y)}
    \sum_{j=1}^m
    \mu_j\tilde w_j(y)<1,\\[1em]
    \text{any value in }[0,1], &
    \displaystyle
    \frac{p(y\mid x)}{p(y)}
    \sum_{j=1}^m
    \mu_j\tilde w_j(y)=1.
    \end{cases}
\end{equation}

Now take the deterministic set \(\mathcal C_{\lambda,t}\) from the
proposition statement and let
\[
    c_{\lambda,t}(x,y)
    =
    \mathbf 1\{y\in\mathcal C_{\lambda,t}(x)\}.
\]
Then \(c_{\lambda,t}\) is a pointwise minimizer of the Lagrangian as given by equation~\eqref{eq:lagrangain_soln_form} with $\mu_j = \frac{\lambda_j}{t}$.

In order to verify optimality, let \(c\) be any feasible relaxed rule.
Since \(c_{\lambda,t}\) minimizes \(\mathcal L(\cdot,\mu)\),
\[
    \mathcal L(c_{\lambda,t},\mu)
    \le
    \mathcal L(c,\mu).
\]
Since \(c\) is feasible,
\[
    (1-\alpha_j)-\mathrm{MacroCov}_{g_j,w_j}(c)\le 0 \ \forall \ j,
\]
and 
\[
    \mathcal L(c,\mu)
    \le
    \mathbb E\left[\sum_{y\in\mathcal Y}c(X,y)\right].
\]
On the other hand, since $\mathcal{C}_{\lambda,t}$ satisfies condition~\eqref{eq:slackness_condition},
\[
    \mathcal L(c_{\lambda,t},\mu)
    =
    \mathbb E|\mathcal C_{\lambda,t}(X)|.
\]
Combining these inequalities gives
\[
    \mathbb E|\mathcal C_{\lambda,t}(X)|
    \le
    \mathbb E\left[\sum_{y\in\mathcal Y}c(X,y)\right]
\]
for every feasible relaxed rule \(c\). Thus \(\mathcal C_{\lambda,t}\) is an optimal deterministic solution to the
simultaneous coverage-constrained problem.
\end{proof}

\section{Additional Experimental Results}
\label{sec:additional_experiments_APPENDIX}

Table \ref{tab:plantnet-trunc_generalized_alpha=0.05} is analogous to Table \ref{tab:plantnet-trunc_generalized_alpha=0.1} in the main text, but for $\alpha = 0.05$. 

\begin{table}[h]
\centering
\setlength{\tabcolsep}{2.2pt}
\caption{Targeting generalized macro-coverage on \truncplant\ at $\alpha = 0.05$. The left panel targets tail-focused macro-coverage; the right panel targets genus-level macro-coverage.}\label{tab:plantnet-trunc_generalized_alpha=0.05}
\begin{tabular}{ll rrr rrr}
\toprule
 & & \multicolumn{3}{c}{\tcbox[on line, colback=green!10, colframe=green!10, arc=3pt, boxsep=1pt, left=3pt, right=3pt, top=1pt, bottom=1pt]{$\mathsf{Goal}: \mathrm{MacroCov}_{\mathrm{tail}} \geq 0.95$}} & \multicolumn{3}{c}{\tcbox[on line, colback=orange!15, colframe=orange!15, arc=3pt, boxsep=1pt, left=3pt, right=3pt, top=1pt, bottom=1pt]{$\mathsf{Goal}: \mathrm{GenusMacroCov} \geq 0.95$}} \\
\cmidrule(lr){3-5} \cmidrule(lr){6-8}
Method & Score & \small{$\mathrm{MarginalCov}$} & \small{$\mathrm{MacroCov}_{\mathrm{tail}}$} & \small{$\mathrm{AvgSize}$} & \small{$\mathrm{MarginalCov}$} & \tiny{$\mathrm{GenusMacroCov}$} & \small{$\mathrm{AvgSize}$} \\
\midrule
\multirow{2}{*}{\standard} & $s_{\softmax}$ & $0.951_{\pm 0.001}$ & $0.785_{\pm 0.002}$ & $5.9_{\pm 0.1}$ & $0.951_{\pm 0.000}$ & $0.939_{\pm 0.001}$ & $5.9_{\pm 0.1}$ \\
 & $\hat{s}_{g,w}$ & $0.949_{\pm 0.001}$ & $0.899_{\pm 0.001}$ & $3.8_{\pm 0.1}$ & $0.950_{\pm 0.000}$ & $\mathbf{0.963}_{\pm 0.000}$ & $4.5_{\pm 0.0}$ \\
\midrule
\classwise & $s_{\softmax}$ & $0.987_{\pm 0.000}$ & $\mathbf{0.996}_{\pm 0.000}$ & $263.4_{\pm 0.9}$ & $0.970_{\pm 0.000}$ & $\mathbf{0.970}_{\pm 0.001}$ & $60.2_{\pm 1.1}$ \\
\midrule
\multirow{2}{*}{\makecell[l]{\textsc{Label-}\\\textsc{weighted}}} & $s_{\softmax}$ & $0.993_{\pm 0.000}$ & $\mathbf{0.950}_{\pm 0.001}$ & $71.5_{\pm 1.9}$ & $0.962_{\pm 0.001}$ & $\mathbf{0.952}_{\pm 0.001}$ & $8.1_{\pm 0.2}$ \\
 & $\hat{s}_{g,w}$ & $0.984_{\pm 0.000}$ & $\mathbf{0.951}_{\pm 0.001}$ & $\mathbf{10.7}_{\pm 0.3}$ & $0.930_{\pm 0.001}$ & $\mathbf{0.952}_{\pm 0.001}$ & $\mathbf{3.4}_{\pm 0.1}$ \\
\bottomrule
\end{tabular}
\end{table}


\end{document}